\documentclass[letterpaper, 10 pt, conference]{ieeeconf}

\IEEEoverridecommandlockouts                              

\usepackage{amsmath}
\usepackage{amssymb}
\usepackage{amsfonts}
\usepackage[capitalize,noabbrev]{cleveref}
\usepackage{subcaption}
\usepackage{graphicx}
\usepackage{algorithm}
\usepackage{algpseudocode}
\newtheorem{Assumption}{Assumption}
\newtheorem{Definition}{Definition}
\newtheorem{Proposition}{Proposition}

\newtheorem{Remark}{Remark}
\usepackage{xcolor}

\newcommand{\resp}[1]{{\color{black}  #1}}

\usepackage{comment}
\title{\LARGE \bf
Computationally efficient robust MPC \\ using optimized constraint tightening
}
\author{Anilkumar Parsi, Panagiotis Anagnostaras, Andrea Iannelli, Roy S. Smith
\thanks{
This work is supported by the Swiss National Science Foundation under grant no. $200021\_178890$. 
The authors are with the Automatic Control Lab, ETH Z\"{u}rich, Switzerland.
Email: {\tt\small panagnost@student.ethz.ch, \{aparsi,iannelli,rsmith\}@control.ee.ethz.ch}}
}

\begin{document}

\maketitle
\thispagestyle{empty}
\pagestyle{empty}

\begin{abstract}

A robust model predictive control (MPC) method is presented for linear, time-invariant systems affected by bounded additive disturbances. The main contribution is the offline design of a disturbance-affine feedback gain whereby the resulting constraint tightening is minimized. This is achieved by formulating the constraint tightening problem as a convex optimization problem with the feedback term as a variable. The resulting MPC controller has the computational complexity of nominal MPC, and guarantees recursive feasibility, stability and constraint satisfaction. The advantages of the proposed approach compared to existing robust MPC methods are demonstrated using numerical examples.


\end{abstract}

\section{INTRODUCTION}\label{introduction}

Model predictive control (MPC) is a powerful control design technique which can naturally handle multiple input multiple output systems and hard constraints on states and inputs while guaranteeing stability of the closed-loop \cite{morariBook}.

Robust MPC methods consider the case where an exact model of the system is not available. The simplest form of uncertainty is represented by an additive disturbance acting on the state dynamics (or process noise). This problem has been extensively studied for linear systems in the robust MPC literature \cite{Kouvaritakis} and a typical approach is to 
split the design problem into two parts. 
In the first part, called the prediction horizon, a sequence of control laws is computed such that the system's state at the end of the prediction horizon reaches a terminal set. Then, the terminal set is designed such that it is an invariant set of the uncertain system under linear feedback and lies within feasible region of the state space. 

In order to achieve tractable formulations, various parameterizations of the control law have been proposed in the literature. One of the early approaches, referred to in the rest of the paper as Tube MPC (TMPC), uses an affine state feedback policy in the prediction horizon \cite{Chisci}. In \cite{Mayne2}, a similar parameterization is used in combination with a state tube. The state tube is a sequence of parameterized sets containing all feasible trajectories of the system. Whereas \cite{Mayne2} computes the evolution of the state tube online, \cite{Chisci} uses a constraint tightening approach to ensure that all future realizations of the state lie within the constraint set. A comparison between these two methods can be found in \cite{Zanon}, where it was shown that the constraint tightening approach in \cite{Chisci} results in a slightly improved performance. 
Various methods have been proposed which use alternative parameterizations of the state feedback controller within the prediction horizon \cite{munoz2014striped}, or use a flexible formulation to represent the state tube \cite{homothetic}.

An alternative parameterization, which alleviates the conservatism associated with a fixed state-feedback policy, was proposed in  \cite{Goulart} in terms of disturbance-affine feedback control actions. Such a parameterization enables online optimization of the disturbance affine feedback gain. This results in much larger regions of attraction compared to strategies such as \cite{Chisci,Mayne2}, where the state feedback term is computed offline. However, 
the size of the online optimization problem grows quadratically with the length of the prediction horizon. 
\resp{The aforementioned robust MPC methods have been combined with reinforcement learning techniques \cite{zanon2020safe}, and successfully implemented on real-world systems \cite{spacecraft}.} Nevertheless, the previously discussed aspects concerning conservatism and scalability remain open problems and thus motivated the current work. \resp{A recent, related work is \cite{sieber2022system}, where conservatism is reduced by parameterizing state and input trajectories of the system using future disturbances.}


 
In this paper, a novel way to design robust MPC controllers for linear time-invariant systems affected by additive disturbances is presented. In order to guarantee robust constraint satisfaction with a computationally efficient program that can be efficiently solved online, a constraint tightening approach is used. Owing to the improved performance of disturbance affine feedback \cite{Goulart} compared to state feedback, a similar parameterization of the control law is used within the prediction horizon. The main contribution of this work is to formulate the constraint tightening as a convex function of the disturbance affine feedback gain. 
This effectively allows offline optimization over the constraint tightening and thus, indirectly, over the size of the region of attraction. The optimized tightened constraints are used to formulate the online MPC problem whose size grows linearly with the length of the prediction horizon. Numerical simulations show that the computation times are one order of magnitude faster compared to the approach presented in \cite{Goulart}, while improving on the region of attraction obtained with TMPC \cite{Chisci}.
 


\paragraph*{Notation}
The $i^{th}$ row of a matrix $A$ is denoted by $[A]_i$. The value of $x$ at the $k^{th}$ time-step is denoted by $x_k$, and $x_{k,i}$ denotes the value of $x$ $i$ steps after the current time-step $k$. 
For a matrix $A$, vector $y$ and a set $Y$ of appropriate dimensions, $\max_{y\in Y}Ay $ refers to row-wise maximization.
The set of real numbers is denoted by $\mathbb{R}$ and the sequence of non-negative integers from $a$ to $b$ is represented by $\mathbb{Z}_{[a,b]}$. The symbol $\geq $ when used for matrices denotes element-wise comparison, $\preccurlyeq$ denotes a conic constraint \resp{and $ \otimes $ denotes the Kronecker product}. 

\section{BACKGROUND}\label{background}
\subsection{Problem formulation}
This work considers discrete-time, linear, time-invariant systems of the following form
\begin{align}\label{system_dynamics}
x_{k+1} = Ax_{k} + Bu_{k} + B_w w_{k},
\end{align}
where $x_k\in \mathbb{R}^{n_x}, u_k \in \mathbb{R}^{n_u}$ and $w_k\in \mathbb{R}^{n_w}$ denote the state of the system, the control input and the additive disturbance respectively at time-step $k$. \resp{ The matrices $A,B$ and $B_w$ are known and no assumptions are made on their structure. Full state feedback is assumed,}  and the state and the input are constrained as
\begin{align}\label{stage_cons}
(x_k,u_k) \in \mathbb{C} := \{(x,u) | Fx + Gu \le b \} ,
\end{align}
where $F\in \mathbb{R}^{n_c \times n_x}, G \in \mathbb{R}^{n_c \times n_u}$ and $b \in \mathbb{R}^{n_c}$.
\begin{Assumption}\label{As:disturbance_assumption}
The additive disturbance $w_k$ lies in a bounded polytope $\mathbb{W}$ which contains the origin, defined as
\begin{align}\label{disturbance_constraint}
    \mathbb{W} \subseteq \mathbb{R}^{n_w}= \{w  |  Dw\leq d\}.
\end{align}
where $D \in \mathbb{R}^{n_d \times n_w}$ and $d\in \mathbb{R}^{n_d}$,
\end{Assumption}

The goal is to design a control policy in order to steer the system governed by \eqref{system_dynamics} from a given initial condition $x_0$, to a region around the origin while guaranteeing that the constraints in  \eqref{stage_cons} are robustly satisfied, i.e., for all noise realizations compatible with \eqref{disturbance_constraint}. 
Moreover, the following cost function is desired to be minimized
\begin{align}\label{eq:InfHorCost}
J = \textstyle\sum_{i=0}^{\infty} {x}_{k}^T Q {x}_{k} + {u}_{k}^T R {u}_{k} .
\end{align}
A MPC algorithm will be used to achieve the above control task. The infinite horizon problem is relaxed by considering a finite prediction horizon of $N$ time-steps, introducing appropriate terminal ingredients,  and solving a receding horizon optimization problem \cite{morariBook}. Because the cost function \eqref{eq:InfHorCost} depends on future disturbances, a disturbance-free prediction is used to formulate the MPC cost function.


\subsection{Existing methodologies}
The existing methodologies in the relevant robust MPC literature differ in the way the control input is parameterized within the prediction horizon. In tube MPC methods \cite{Chisci} \cite{Mayne2}, the prediction of the future states $x_{k,i}$ is split into nominal $\hat{x}_{k,i}$ and error $e_{k,i}$ components such that $e_{k,i}=x_{k,i}-\hat{x}_{k,i}$. \resp{ In \cite{Chisci}, the parameterization $u_{k,i} = \hat{u}_{k,i}+Ke_{k,i},\: \: i\in\mathbb{Z}_{[0,N-1]}$ is used to obtain the following dynamics for $i\in\mathbb{Z}_{[0,N-1]}$},
\begin{subequations}\label{eq:Split_Dyn}
\begin{align}
    \hat{x}_{k,i+1} = A \hat{x}_{k,i} + B \hat{u}_{k,i}, \quad \hat{x}_{k,0} = x_k,\label{eq:Split_Dyn1} \\
    e_{k,i+1} = (A+BK)e_{k,i} + B_w w_{k,i}, \quad e_{k,0} = 0, \label{eq:Split_Dyn2}
\end{align}
\end{subequations}
where $K\in \mathbb{R}^{n_u\times n_x}$ is a precomputed feedback gain and $\hat{u}_{k,i}$ are online optimization variables. 
The state at the end of the prediction horizon is driven into a terminal set, where the control law $u_{k,i} {=} K x_{k,i}$ is always feasible and \resp{ensures forward invariance of the terminal set. Because the dynamics in \eqref{eq:Split_Dyn2} are completely determined by  $A,B$ and $K$, the worst-case evolution of the error components are used to tighten the constraint set $ \mathbb{C} $ offline, thereby ensuring constraint satisfaction under any possible disturbance sequence in the future. The method in \cite{Mayne2} computes a robust invariant set for the dynamics \eqref{eq:Split_Dyn2} which is then used to tighten $ \mathbb{C} $. } However, both these strategies lead to a reduction in the region of attraction (ROA) of the closed loop system, and there is not a systematic way of choosing $K$ to reduce the conservatism.

An alternative approach which yields larger ROA is presented in \cite{Goulart}, which uses a disturbance affine feedback gain to parameterize the control input, as $ {u}_{k,0} {=} \hat{u}_{k,0} $ and ${u}_{k,i} = \hat{u}_{k,i} + \sum_{j=0}^{i-1} M_{i,j} B_w {w}_{k,j}$, for $i{\in}\mathbb{Z}_{[1,N-1]}$ \resp{ where $\hat{u}_{k,i}$ and $M_{i,j}\in \mathbb{R}^{n_u\times n_x}$ are online optimization variables. This method is referred to as fully parameterized disturbance affine MPC (FPD) in the remainder of this paper.}

The input parameterization used in TMPC and FPD are related to each other, as shown below. Define $\mathbf{u}_{k} = [u_{k,0}^T,u_{k,1}^T,\ldots,u_{k,N-1}^T]^T$, $\mathbf{\hat{u}}_{k} = [\hat{u}_{k,0}^T,\hat{u}_{k,1}^T,\ldots,\hat{u}_{k,N-1}^T]^T$ and $\mathbf{w}_k = [w_{k,0}^T,w_{k,1}^T,\ldots,w_{k,N-1}^T]^T$.  Then, the control laws for TMPC and FPD can be compactly written as 
 \begin{align}\label{eq:H_val}
 \begin{split}
    \mathbf{u}_{k} = \mathbf{\hat{u}}_k + \mathbf{H} \mathbf{B_w} \mathbf{w}_k, \; \;
    \mathbf{H} = 
    \begin{cases}
    \mathbf{K} \quad &\textnormal{for TMPC \cite{Chisci}}, \\
    \mathbf{M} \quad &\textnormal{for FPD \cite{Goulart}},
    \end{cases}
 \end{split}
\end{align}
where $\mathbf{B_w} = I_N \otimes B_w $.  
In \eqref{eq:H_val}, the gain $\mathbf{K}$ is given by
\begin{align}\label{bold_K}
    \mathbf{K}:=\begin{bmatrix}
        0 & ... & ... & ...& 0 \\ K & 0 & ...& ... & 0 \\ K(A+BK) & K &... & ... & 0  \\ \vdots & \vdots & \ddots & \vdots & \vdots\\K(A+BK)^{N-2} & ... & ...&K & 0
        \end{bmatrix},
\end{align}
and $\mathbf{M}$ is also a strictly lower triangular matrix where the non-zero blocks are the individual feedback gains $M_{i,j}$. 
Because the latter are chosen online, the control parameterization used in FPD is much more flexible than TMPC. This results in improved closed loop performance and a larger ROA compared to TMPC. This comes at an increase in the required computational effort. By fixing  $\mathbf{K}$ offline, the computational complexity of TMPC is the same as that of nominal MPC. In contrast, the number of optimization variables in FPD grows quadratically with the length of the prediction horizon and the number of state variables. Another recent method, called fusion of tubes MPC \cite{Kogel} tries to improve the performance of TMPC by choosing the feedback gain as a linear combination of a finite number of predefined gains $K^{[j]}$. 
However, a systematic design procedure of the predefined gains is not provided.

\section{Offline design for constraint tightening}
\label{Sec:DistFeedback}
A disturbance affine feedback policy is used for control in light of its \resp{ability to parameterize closed loop trajectories as optimization variables \cite{Goulart}}. In this section, the relationship between the disturance feedback gain and the resulting constraint tightening will be established. In addition, design conditions for the terminal controller and terminal set will be formulated. 

\subsection{Input parameterization and tightened constraint sets}
The control input is parameterized as
\begin{align}\label{PM_input}
\mathbf{u}_k = \hat{\mathbf{u}}_k+\mathbf{M}_{\text{off}} \mathbf{B_w}\mathbf{w}_k
\end{align}
where $\hat{\mathbf{u}}_k$ are online optimization variables and $\mathbf{M}_{\text{off}}$ is a matrix computed offline. In order to ensure causality and recursive feasibility by design, $\mathbf{M}_{\text{off}}$ satisfies the structure
\begin{align}\label{causality}
    \begin{split}
        \mathbf{M}_{\text{off}}=\begin{bmatrix}
        0 & 0& 0 &...&  0 \\ M_{1}^{\text{off}} & 0 & 0& ... & 0 \\ M_{2}^{\text{off}} & M_{1}^{\text{off}} &0&... & 0  \\ \vdots & \vdots & \ddots & \vdots & \vdots \\M_{N-1}^{\text{off}} & ... & ... & M_{1}^{\text{off}} & 0
        \end{bmatrix}
    \end{split}
\end{align}
where $M^{\text{off}}_{i} {\in} \mathbb{R}^{n_u {\times} n_x}$. Let $\mathcal{M}$ denote the set of matrices satisfying \eqref{causality}. \resp{Because  $M_{1}^{\text{off}},\ldots,M_{N-1}^{\text{off}}$ can be independently chosen, \eqref{causality} is less restrictive compared to \eqref{bold_K}.}


Under the proposed feedback law \eqref{PM_input}, the state dynamics follow \eqref{eq:Split_Dyn1} and the error dynamics can be written as
\begin{align}
    e_{k,i+1} &= Ae_{k,i}+ B\textstyle\sum_{j=0}^{i}M_{i\resp{-j}}^{\text{off}} B_w w_{k,j} , \label{eq:Split_Dyn_PM1}
\end{align}
where $e_{k,0} = 0$, $i\in \mathbb{Z}_{[0,N-1]}$ and $BM_{0}^{\text{off}}:=I_{n_x}$ is a notational overload used for simplicity. The predicted error can be compactly written for $i\in\mathbb{Z}_{[0,N-1]}$ as
\begin{align}\label{eq:CompactError}
\begin{split}
    e_{k,i+1} &= \textstyle\sum_{l=0}^{i} \bigr(\textstyle\sum_{j=0}^{i-l} A^j B M^{\text{off}}_{i-l-j} \bigr) B_w w_{k,l}. 
\end{split}
\end{align}
For $i{\in}\mathbb{Z}_{[0,N-1]} $, define $\mathbb{E}_{i} := \{(e_{k,i},u_{k,i}-\hat{u}_{k,i})|\: \eqref{eq:CompactError}, \: w_{k,l}\in\mathbb{W}, \: \forall l\in \mathbb{Z}_{[0,i-1]} \}$. Using this sequence of sets, tightened constraints can be imposed on the nominal state and input sequences. 
 For this purpose, define $ \forall i\in \mathbb{Z}_{[0,N-1]} $
\begin{equation}\label{eq:Tightened_Sets}
    \mathbb{C}_i := \mathbb{C} \ominus \mathbb{E}_i = \{(x,u) | Fx + Gu \leq b -t_{i}\}, 
\end{equation}
where $t_i{\in}\mathbb{R}^{n_c}$ represents the tightening applied on $\mathbb{C}$ $i$ time-steps into the future. The predicted trajectories in the online optimization problem must thus satisfy $(\hat{x}_{k,i},\hat{u}_{k,i}) \in \mathbb{C}_i$ for all $i\in\mathbb{Z}_{[0,N-1]}$. 

In order to compactly represent the tightened state and input constraints, the following notation is introduced. Let $\mathbf{x}_k=[x_{k,0}^T,x_{k,1}^T,...,x_{k,N-1}^T]^T$, $\mathbf{t} = [t_0^T,\ldots,t_{N-1}^T]^T$ and $\mathcal{W}=\mathbb{W}\times \mathbb{W} \times...\times \mathbb{W}$. Substituting the control law \eqref{PM_input} into dynamics \eqref{system_dynamics}, the state of the system can be written as

\begin{align}\label{PM_sys_ev}
 \mathbf{x}_k = \mathbf{C}_{xx} x_k + \mathbf{C}_{xu} \hat{\mathbf{u}}_k+ (\mathbf{C}_{xw}+\mathbf{C}_{xu} \mathbf{M}_{\text{off}})\mathbf{B}_w \mathbf{w}_k,
\end{align}
where
\begin{footnotesize}
\begin{align*}
    &\mathbf{C}_{xx} := \begin{bmatrix}
    I, A,\hdots, A^{N-1}
    \end{bmatrix}^T, \quad \mathbf{B}_w = I_{N-1} \otimes B_w,\\
    &\mathbf{C}_{xu} {:=} \begin{bmatrix}
    0 & ... & ... & 0 \\ B & 0 & ... & 0 \\ AB & B & ... & 0 \\ \vdots & \vdots & \ddots & \vdots \\ A^{N{-}2} B &...&AB &  B
    \end{bmatrix},
    \mathbf{C}_{xw} {:=} \begin{bmatrix}
    0 & ... & ... & 0 \\ I & 0 & ... & 0 \\ A & I & ... & 0\\ \vdots & \vdots & \ddots & \vdots \\ A^{N{-}2} & ... & ... & I
    \end{bmatrix}.
\end{align*}
\end{footnotesize}
By denoting $\mathbf{F} := I_{N} \otimes F$, $
    \mathbf{G} := I_{N} \otimes G $, $  
    \mathbf{b} := 1_{N} \otimes  b $, $ \mathbf{D}:=I_{N-1} \otimes D$, and $
    \mathbf{d} := \mathbf{1}_{N-1} \otimes d $, the constraints on the system trajectory in the prediction horizon of MPC are
\begin{align}\label{eq:aug_con}
    \mathbf{F} \mathbf{x}_k + \mathbf{G} \mathbf{u}_k \le \mathbf{b}.
\end{align}
Using \eqref{PM_sys_ev}, \eqref{eq:aug_con} can be written as
\begin{subequations}
\begin{align}
&(\mathbf{F}\mathbf{C}_{xu} + \mathbf{G}) \hat{\mathbf{u}}_k + \mathbf{F} \mathbf{C}_{xx} x_k \leq \mathbf{b}- \mathbf{t}, \label{tightened_cons} \\
\mathbf{t}&=\max_{\mathbf{w}_k \in \mathcal{W}} (\mathbf{F}\mathbf{C}_{xw} + \mathbf{F}\mathbf{C}_{xu} \mathbf{M}_{\text{off}} + \mathbf{G} \mathbf{M}_{\text{off}})\mathbf{B}_w \mathbf{w}_k. \label{t}
\end{align}
\end{subequations}
Note that the maximization in \eqref{t} is performed row-wise. For $j\in \mathbb{Z}_{[1,Nn_c]}$, $[t]_j$ represents the tightening to be performed on the $j^{th}$ constraint. Using strong duality, the $j^{th}$ row in \eqref{t} can be equivalently written as 
\begin{subequations}\label{i_opt_prob_2}
\begin{align}
    \min_{z_{(j)},[t]_j} \quad  &[t]_j  \\
     \text{s.t.} \quad [t]_j &=z_{(j)}^T \mathbf{d},\quad z_{(j)} \geq 0,\quad [t]_j \geq 0, \\ 
    z_{(j)}^T \mathbf{D} &= [(\mathbf{F}\mathbf{C}_{xw} + \mathbf{F}\mathbf{C}_{xu} \mathbf{M}_{\text{off}} + \mathbf{G} \mathbf{M}_{\text{off}})\mathbf{B}_w]_j.
\end{align}
\end{subequations}
Thus, for a fixed  $\mathbf{M}_{\text{off}}$, \eqref{t} can be formulated as $Nn_c$ linear programs in the variables $z_{(j)}$ and $[t]_j$. \resp{The key observation leveraged here is that}, if $\mathbf{M}_{\text{off}}$ is an optimization variable, it can be chosen such that the tightenings in $\mathbf{t}$ are minimized.

\subsection{Terminal sets} 
A terminal set is to be chosen such that a linear feedback policy of the form $u = K_f x$ is always feasible inside the terminal set. In order to ensure recursive feasibility of the online optimization problem, the following assumption needs to be satisfied by the terminal set and the terminal controller.

\begin{Assumption}\label{As:TerminalSet}
There exists a terminal set $\mathcal{X}_T$, a disturbance feedback gain $M_N^\text{off} \in \mathbb{R}^{n_u \times n_x} $ and a terminal feedback gain $K_f \in \mathbb{R}^{n_u\times n_x}$ such that for all $x\in\mathcal{X}_T$ and $w \in \mathbb{W}$,
\begin{subequations}\label{eq:TermSetAssump}
\begin{align}
 &\bigr(x{+} \textstyle\sum_{j=0}^{N-1}A^j BM_{N{-}1{-}j}^{\text{off}} B_w w, \: K_fx {+} M_N^\text{off}B_w w \bigr)     
     \in \mathbb{C}_{N{-}1},  \label{eq:TermSetAssump1} \\
&(A+BK_f)x {+} \textstyle\sum_{j=0}^{N} A^j B M_{N{-}j}^{\text{off}}B_w  w  \in \mathcal{X}_T. \label{eq:TermSetAssump2}
\end{align}
\end{subequations}
\end{Assumption}


\begin{Remark}\label{Rem:RecFeasCond}
Note that conditions similar to \eqref{eq:TermSetAssump} are imposed on the terminal set to ensure recursive feasibility of TMPC in \cite{Chisci,Kogel}. 
\end{Remark}

It can be seen that Assumption \ref{As:TerminalSet} results in a coupling of the design of the terminal set $\mathcal{X}_T$, the terminal feedback gain $K_f$ and the disturbance affine feedback gains $\mathbf{M}_\text{off}$ and $M_N^\text{off}$. Even when the feedback gains are known, the design of terminal sets is a difficult problem, since they are computed by iterative set-intersections performed on polytopic sets \cite{Kouvaritakis}. In this work, the design process is simplified by first choosing the terminal components as is done in standard tube MPC methods \cite{Chisci}. The gain $K_f$ is chosen such that $(A+BK_f)$ is Schur stable, and the terminal set is constructed as a polytope of the form $\mathcal{X}_T := \{x| Yx \le z\}$. Then, \eqref{eq:TermSetAssump} will be reformulated as convex constraints on the disturbance affine feedback gains $\mathbf{M}_\text{off}$ and $M_N^\text{off}$ in the following proposition.
%

\begin{Proposition}\label{Prop:TermSetDesign}
Given the terminal ingredients $\mathcal{X}_T$ and $ K_f$, define $c_F = \max_{x\in\mathcal{X}_T} (F+GK_f) x$ and $c_Y =\max_{x\in\mathcal{X}_T} Y(A + BK_f) x $. If there exist $\Lambda_1 \in  \mathbb{R}^{n_c \times n_d}$ and $\Lambda_2 \in  \mathbb{R}^{n_t \times n_d}$, such that
 \begin{align}\label{eq:TermSetDesign_Cond}
 \begin{split} 
     \Lambda_1 \ge 0, \quad \Lambda_2 \ge 0,   \quad  \Lambda_1 d  + c_F \le b &- t_{N-1} ,  \\
     F\textstyle\sum_{j=0}^{N-1}A^j BM_{N{-}1{-}j}^{\text{off}} B_w + GM_N^\text{off} B_w &= \Lambda_1 D , \\
     Y \textstyle\sum_{j=0}^{N} A^j B M_{N{-}j}^{\text{off}}B_w  = \Lambda_2 D,\quad   \Lambda_2 d &+  c_Y \le z, 
 \end{split}
 \end{align}
 then Assumption \ref{As:TerminalSet} is satisfied.
\end{Proposition} 
\begin{proof}
For a given $\mathcal{X}_T$ and $K_f$, \eqref{eq:TermSetAssump1} is satisfied if
\begin{align}\label{eq:TermSet_cond1}
   \max_{ w\in\mathbb{W}} \: F \textstyle\sum_{j=0}^{N-1}A^j &BM_{N{-}1{-}j}^{\text{off}}B_w w + G M_N^\text{off}B_w w + \nonumber \\ 
  &\max_{x\in\mathcal{X}_T}  Fx+ GK_f x   \le b-t_{N-1}. 
\end{align}
Similarly, \eqref{eq:TermSetAssump2} is satisfied if 
\begin{align}\label{eq:TermSet_cond2}
\max_{ w\in\mathbb{W}} \: \sum_{j=0}^{N} A^j B M_{N{-}j}^{\text{off}}B_w  w  + \max_{x\in\mathcal{X}_T}  Y(A+BK_f) x   \le z. 
\end{align}
In  \eqref{eq:TermSet_cond1} and \eqref{eq:TermSet_cond2}, the maximizations over $\mathcal{X}_T$ can be performed to compute $c_F, c_Y$ before choosing $\mathbf{M}_\text{off}$ and $M_N^\text{off}$, as the terminal ingredients are known. Then, applying strong duality to the maximization over $\mathbb{W}$, \eqref{eq:TermSet_cond1} and \eqref{eq:TermSet_cond2} are satisfied iff there exist $\Lambda_1$ and $\Lambda_2$ such that \eqref{eq:TermSetDesign_Cond} holds.
\end{proof}

\section{Robust MPC controller}
This section presents the algorithm defining the robust MPC controller proposed in this work. First, the offline and online optimization problems to be solved are described. Then, the system theoretic properties of the closed loop system are presented.


\subsection{Offline optimization}
Building on the formulation illustrated in Section \ref{Sec:DistFeedback} of convex constraints on  $\mathbf{M}_\text{off}$ and $M_N^\text{off}$, these design matrices can be computed as the solution of an optimization problem. In order to reduce the conservatism of the robust MPC controller, it is desired that the feasible region of the online optimization problem is large. This feasible region is represented by a polytope in a high dimensional space, whose volume is difficult to compute. One way to approximate this objective is to  minimize the amount of constraint tightening to be performed. To this aim, the following optimization problem is solved
\begin{subequations}\label{PM_l2_norm}
\begin{align}
    &\min_{\substack{Z,\mathbf{t},\Gamma,\mathbf{M}_{\text{off}},\\
    M^{\text{off}}_N,\Lambda_1,\Lambda_2}} \quad  \|\mathbf{t}\|_2 \label{PM_l2_norm1} \\
    \text{s.t.} \quad &Z^T \mathbf{D}=(\mathbf{F}\mathbf{C}_{xw} + \mathbf{F}\mathbf{C}_{xu} \mathbf{M}_{\text{off}} + \mathbf{G} \mathbf{M}_{\text{off}})\mathbf{B}_w, \label{PM_l2_norm2} \\ 
    &Z^T \mathbf{d}= \mathbf{t}, \quad Z \geq 0, \quad \mathbf{t} \ge 0, \label{PM_l2_norm3}\\
    &\mathbf{b} - \mathbf{t} \geq 0 ,\quad  \mathbf{M}_{\text{off}} \in \mathcal{M}, \quad \eqref{eq:TermSetDesign_Cond},\label{PM_l2_norm5}
\end{align}
\end{subequations}
where $Z= [
z_{(1)}, \hdots, z_{(Nn_c)}
]$.
In problem \eqref{PM_l2_norm}, \eqref{PM_l2_norm2}-\eqref{PM_l2_norm3} are obtained by stacking the individual constraints from \eqref{i_opt_prob_2}. Note that $\mathbf{M}_{\text{off}}$ is an optimization variable in \eqref{PM_l2_norm}.
The \resp{first constraint in \eqref{PM_l2_norm5} is introduced to preserve the feasibility of the online MPC optimization problem}. Using an epigraph reformulation, \eqref{PM_l2_norm} can be transformed into a convex second-order cone program  which can be efficiently solved \cite{boyd_vandenberghe_2004}.
\begin{Remark}\label{Rem:Fusion}
Problem \eqref{PM_l2_norm} provides a formal way to design constraint tightening by optimizing offline over the disturbance feedback gains. Following a similar approach to \cite{Kogel}, the proposed methodology can be easily extended to fuse multiple feedback gains online. If this strategy is pursued, additional feedback gains can be designed, for example, by using a weighted norm in the objective function of \eqref{PM_l2_norm}.
\end{Remark}
\begin{Remark}\label{Rem:Kf_better}
Problem \eqref{PM_l2_norm} minimizes  $\|\mathbf{t}\|_2$ as a surrogate for the maximization of the feasible region. 
Note that the TMPC solution \cite{Chisci} is feasible for \eqref{PM_l2_norm}, which can be verified by choosing $\mathbf{M}_\text{off}$ to be of the structure in \eqref{bold_K} with $K=K_f$ and $M_N^\text{off} = K_f(A+BK_f)^{N-1}$. In simulation studies, it was found that the proposed method could sometimes lead to smaller ROA compared to TMPC \resp{using $K_f$ for feedback}. This might be justified considering that minimizing $\|\mathbf{t}\|_2$ is a heuristic to maximize the size of the feasible region. A possible remedy for this problem, which can be detected offline, is to explicitly enforce that the tightenings in $\mathbf{t}$ are smaller than those resulting from TMPC. 
\end{Remark}

\subsection{Receding horizon control}
Once the constraint tightenings are computed from \eqref{PM_l2_norm}, the online optimization problem 
can be formulated as
\begin{subequations}\label{OptProb_PM}
    \begin{align}
        \min_{\hat{\mathbf{u}}_k} \quad \textstyle\sum_{i=0}^{N-1}(\hat{x}_{k,i}^\top Q \hat{x}_{k,i}  &{+} \hat{u}_{k,i}^\top R \hat{u}_{k,i})\label{OptProb_PM1} {+} \hat{x}_{k,N}^\top P \hat{x}_{k,N}\\
		\text{s.t.} \quad A\hat{x}_{k,i} + B\hat{u}_{k,i} &= \hat{x}_{k,i+1}, \quad \hat{x}_{k,0}=x_k, \label{OptProb_PM2}\\
            F\hat{x}_{k,i} + G\hat{u}_{k,i}  &\le b - t_i,  \quad i\in \mathbb{Z}_{[0,N-1]}, \label{OptProb_PM3} \\
            Y\hat{x}_{k,N} &\le z . \label{OptProb_PM4}
    \end{align}
\end{subequations}
In \eqref{OptProb_PM}, 
the terminal cost is defined using a positive definite matrix $P$ chosen such that 
\begin{align}\label{eq:TermCostLyap}
    (A+BK_f)^\top P (A+BK_f) + Q + K_f^\top R K_f \preccurlyeq P.
\end{align}
It is worth observing that the computational complexity of \eqref{OptProb_PM} is the same as that of nominal MPC, and thus grows only linearly with the length of the prediction horizon. The robust MPC design scheme is summarized in \cref{pseudoalgo}. 

\setlength{\textfloatsep}{0pt}

\begin{algorithm}[t]
\caption{ Optimized constraint tightening for robust MPC}\label{pseudoalgo}
\begin{algorithmic}[1]
\Statex \textbf{Offline:} 
\State Choose $K_f$ such that $(A+BK_f)$ is Schur stable
\State Design $\mathcal{X}_T$ following \cite{Chisci}, compute $c_F$ and $c_Y$
\State Solve  \eqref{PM_l2_norm} to compute $\mathbf{t}$
\setcounter{ALG@line}{0}
\Statex \textbf{Online:} At each time-step $k\ge0$:
\State Obtain the measurement $x_k=x_{k,0}=\hat{x}_{k,0}$
\State Solve \eqref{OptProb_PM}
\State Apply $\pi(x_k)=\hat{u}^*_{k,0}\in \mathbb{R}^{n_u}$
\end{algorithmic}
\end{algorithm}

\subsection{Closed-loop properties}\label{sec:ClosedLoopProperties}
In this section, it will be shown that problem \eqref{OptProb_PM} is recursively feasible and that the closed loop system is input-to-state stable (ISS). The set of admissible policies and the feasible set of \eqref{OptProb_PM} are first defined.

\begin{Definition}[Admissible control inputs]\label{def:disturb affine admissible policies}
Given a  state $x_k$, $\Pi_W(x_k)$ is the set of all admissible sequences $\hat{\mathbf{u}}_k$ in \eqref{OptProb_PM}. That is, $\Pi_W(x_k):=\left\{ \hat{\mathbf{u}}_k \: | \:  \eqref{OptProb_PM2},\eqref{OptProb_PM3},\eqref{OptProb_PM4} \right\}$.
\end{Definition}

\begin{Definition}[Feasible set]\label{def:disturb affine feasible set}
 The feasible set of the optimization problem \eqref{OptProb_PM} is defined as the set of all initial states for which there exists at least one admissible control policy, that is,  $X_W := \{x_0  \rvert \quad \Pi_W(x_0) \neq \emptyset\}$.
\end{Definition}

\begin{Proposition}[Recursive Feasibility]\label{Prop:Recursive feasibility}
Let  Assumptions \ref{As:disturbance_assumption} and \ref{As:TerminalSet} be satisfied and the offline optimization problem \eqref{PM_l2_norm} have a feasible solution. Then, the receding horizon policy generated by \cref{pseudoalgo} and applied to system \eqref{system_dynamics} results in a recursively feasible online optimization problem \eqref{OptProb_PM}.
\end{Proposition}
\begin{proof}
Assume $x_k \in X_W$, and let the optimal solution to \eqref{OptProb_PM} at time-step $k$ be given by the nominal state and control input sequences $\{ \hat{x}_{k,i}^*\}_{i=0}^{N},\{ \hat{u}_{k,i}^*\}_{i=0}^{N-1} $. The state of the system at time-step $k+1$ is given by $x_{k+1} = Ax_k + B\hat{u}_{k,0}^* + B_w w_k = \hat{x}_{k,1}^* + B_w w_k$.  
The online optimization problem \eqref{OptProb_PM} is recursively feasible if there exists a feasible state and control input sequence at time-step $k+1$. 

Consider the candidate state and input sequences given as
\begin{align}
    &\hat{u}_{k+1,i} =  \hat{u}_{k,i+1}^* + M^{\text{off}}_{i+1}B_w w_k, \quad  i\in \mathbb{Z}_{[0,N-2]}, \nonumber \\
    &\hat{u}_{k+1,N-1} = K_f \hat{x}_{k,N}^* + M_{N}^\text{off}B_w w_k,  \label{eq:Candidate_traj} \\
    &\hat{x}_{k+1,i} = \hat{x}_{k,i+1}^* + \textstyle\sum_{j=0}^{i} A^j B M^{\text{off}}_{i-j} B_w w_k, \quad i\in \mathbb{Z}_{[0,N-1]}, \nonumber  \\
    &\hat{x}_{k+1,N} = (A+B K_f)\hat{x}_{k,N}^* + \textstyle\sum_{j=0}^{N} A^j B M^{\text{off}}_{N-j} B_w w_k.\nonumber 
\end{align}
It can be easily verified that the candidate solution \eqref{eq:Candidate_traj} satisfies the dynamics constraint \eqref{OptProb_PM2} in the online optimization problem. Moreover, for $ i\in \mathbb{Z}_{[0,N-2]}$, the terms in the state and input constraints can be written as
\begin{align}\label{eq:Ineq_recFeas}
    F\hat{x}_{k+1,i} + G\hat{u}_{k+1,i}  &= F \hat{x}_{k,i+1}^* + G\hat{u}_{k,i+1}^* + \\ F\textstyle\sum_{j=0}^{i} &A^j B M^{\text{off}}_{i-j} B_w w_k + G M^{\text{off}}_{i+1}B_w w_k .\nonumber
\end{align}
The first two terms on the right hand side in \eqref{eq:Ineq_recFeas} can be upper bounded by $b-t_{i+1}$, because the trajectory computed at time $k$ satisfies \eqref{OptProb_PM2}. Moreover, the following relationship holds between $t_i$ and $t_{i+1}$ for $ i\in \mathbb{Z}_{[0,N-2]}$,
\begin{align}\label{eq:t_i_i1}
    t_{i+1} &=  \max_{\{w_l\in\mathbb{W}\}_{l=0}^{i}} F \textstyle\sum_{l=0}^{i} \sum_{j=0}^{i{-}l} A^j B M_{i{-}j}^{\text{off}} B_w w_{l} \nonumber\\
    &\qquad + G \textstyle\sum_{l=0}^{i} M_{l+1}^{\text{off}} B_w w_{i-l}  \\
    &= t_i + \max_{w_0\in\mathbb{W}} F\textstyle\sum_{j=0}^{i} A^j B M_{i{-}j}^{\text{off}} B_w w_{0} + G M_i^{\text{off}} B_w w_0. \nonumber
\end{align}
Thus, using \eqref{eq:Ineq_recFeas} and \eqref{eq:t_i_i1}, \eqref{OptProb_PM3} holds for $ i\in \mathbb{Z}_{[0,N-2]}$ at time-step $k+1$. Similarly, the feasibility of \eqref{OptProb_PM3} for $i{=}N{-}1$ and the terminal constraints \eqref{OptProb_PM4} is a direct consequence of the proposed design satisfying \eqref{eq:TermSetAssump1} and  \eqref{eq:TermSetAssump2}, respectively. 
Thus, \eqref{OptProb_PM} is recursively feasible.
\end{proof}

\begin{Proposition}[Input-to-state stability]
Let Assumptions \ref{As:disturbance_assumption} and \ref{As:TerminalSet} be satisfied and the offline optimization problem \eqref{PM_l2_norm} have a feasible solution. Then, the closed loop formed by system \eqref{system_dynamics} and the receding horizon policy $\pi(x_k)$ is ISS. That is, there exists a $\mathcal{K L}$-function $\beta(\cdot)$ and a $\mathcal{K}$-function $\gamma(\cdot)$ as defined in \cite{iss}, such that for $\bar{w} = \sup_{\tau\in \mathbb{Z}_{[0,k-1]}} \|w_\tau \|  $,
 \begin{align}
     &\| x_k\| \leq \beta (\|x_0\|,k) + \gamma (\bar{w}),  \: \forall k \in \mathbb{Z}_{[0,\infty)},  x_0 \in X_W.
 \end{align}
\end{Proposition}
\begin{proof}
The proof follows the approach presented in \cite{Goulart}, where it is shown that the optimal cost function $J^*_N(x)$ is a ISS Lyapunov function \cite{iss} for the closed loop system. Let $f(x_k,w_k)=Ax_k+B\pi(x_k)+B_w w_k$ describe dynamics of the closed loop system. As a consequence of \cref{Prop:Recursive feasibility}, $X_W$ is a RPI set of $f(\cdot)$. 

First, it can be seen that $J^*_N(\cdot)$ is a Lyapunov function for the undisturbed system. This is because, by using Proposition 17 from \cite{Goulart}, it holds that $J^*_N(\cdot)$ and $\pi(\cdot)$ are Lipschitz continuous on $X_W$. It follows from Lemma 4.3 in \cite{Khalil} that there exist $\mathcal{K}_{\infty}$ functions  $a_2(\cdot)$ and $a_3(\cdot)$ such that $a_2(\|x\|)\leq J_N^*(x)\leq a_3(\|x\|)$. Finally, using the candidate solution in \eqref{eq:Candidate_traj}, the optimal value function satisfies $J^*_N(f(x,0))-J^*_N(x) \leq -a_1(\|x\|)$  for a $\mathcal{K}_{\infty}$ function $a_1(\cdot)$.

Now let $L_{J}$ and $L_f$ be the Lipschitz constants of $J^*_N(\cdot)$ and $f(x,\cdot)$, respectively. Then, 
\begin{align}
    &J^*_N(f(x,w))-J_N^*(x)  \nonumber \\
    &= J^*_N(f(x,0))-J^*_N(x)+J^*_N(f(x,w))-J^*_N(f(x,0))  \nonumber\\
        &\le -a_1(\|x\|) + L_{J}L_f\|w\|.
\end{align}
Thus, $J^*_N(x)$ is a ISS Lyapunov function for the closed loop system, and using Lemma 3.5 from \cite{iss}, the closed loop system is ISS with region of attraction $X_W$.
\end{proof}

\section{NUMERICAL EXAMPLES}\label{numerical_examples}
In this section, the proposed optimized constraint tightening (OCT) algorithm is compared to Tube MPC (TMPC) \cite{Chisci} and the fully parameterized disturbance-affine MPC (FPD) \cite{Goulart} methods. 
The algorithms are implemented in MATLAB using MOSEK \cite{mosek}, YALMIP \cite{yalmip} and MPT3 \cite{MPT3}. The reported statistics correspond to the time that the solver requires to solve the optimization problems on a AMD EPYC 7H12 processor with 6GB RAM. The code for simulating these examples is available in the online repository \cite{repo_ETH}. 

Two simple mechanical spring-mass systems are used as working examples. System 1 consists of a mass connected to a fixed end by a spring and a damper. System 2 consists of three masses connected along a line using springs and dampers, with the first mass connected to a fixed end similar to System 1. The dynamics of System 1 can be described by
\begin{small}
\begin{align}\label{Sys1CT}
    \frac{d}{dt}x = \begin{bmatrix}
        0 & 1 \\ -k & -b
    \end{bmatrix}
x   
    +
    \begin{bmatrix}
        0 \\ 20
    \end{bmatrix} u
    + \begin{bmatrix}
        1 & 0 \\ 0 & 1
    \end{bmatrix} w
\end{align}
\end{small}
where the spring constant $k=1\mathrm{Nm^{-1}}$, the damping coefficient $b=0.1\mathrm{Nsm^{-1}}$. The dynamics are discretized using forward Euler method with a sampling time of $T=0.1 \mathrm{s}$. The state, the input and the disturbance of System 1 are subject to linear inequality constraints described by
\begin{align*}
\|x\|_{\infty} \le 25,  \|u\|_{\infty} &\le 1, -2 \le [w]_1 \le 2,   -5 \le [w]_2 \le 5.
\end{align*}
The parameters and constraints on System 2 are similarly defined. Moreover, the cost matrices are given as $Q=I_{n_x} $ and $ R=I_{n_u}$. The terminal feedback gain $K_f$ is computed as the solution to the infinite horizon linear quadratic regulator problem and $P$ is chosen such that it satisfies \eqref{eq:TermCostLyap} with an equality. The terminal set for the TMPC and OCT controllers is computed as suggested in \cite{Chisci}, and the terminal set of the FPD controller is chosen as the maximal robust positively invariant set under the terminal controller \cite{Kouvaritakis}.

\setlength{\textfloatsep}{5pt}
\begin{figure}
        \centering
        \includegraphics[scale=1]{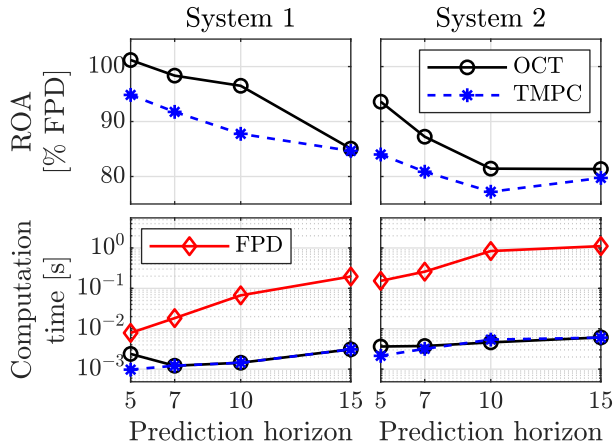}
        \caption{Comparison of average online computational times and ROA for OCT, TMPC and FPD controllers.}
        \label{Fig:all_together}
\end{figure}
\setlength{\textfloatsep}{5pt}
\begin{figure}
	\centering
	\includegraphics[scale=1]{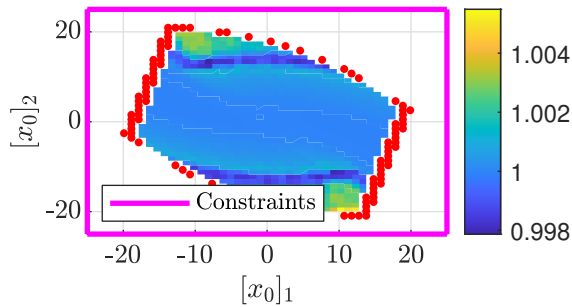}
	\caption{Colorbar shows the ratio of averaged closed loop costs of TMPC and OCT as a function of the initial state for System 1 with $N=10$. The red dots indicate initial states where TMPC is infeasible but OCT is feasible.}
	\label{Fig:perf}
\end{figure}
To compare the performance of the control algorithms, the ROA and the computational times of each controller for Systems 1 and 2 are plotted in Figure \ref{Fig:all_together}. For System 1, the ROA is estimated by dividing the state space into a grid of 2500 uniformly distributed points. The number of points for which the online optimization problem is feasible for TMPC and OCT controllers is shown as a percentage of the number of points for which the FPD controller was feasible.  For System 2, this analysis was performed using a grid consisting of 125,000 points on the $[x]_2-[x]_4-[x]_6$ subspace, by setting all the other states to zero. It can be seen that the proposed method results in a larger ROA compared to the TMPC approach. It must be noted that for System 1 with $N=15$, OCT using \eqref{PM_l2_norm} resulted in a smaller ROA compared to TMPC. However, imposing the additional constraint proposed in Remark \ref{Rem:Kf_better} allowed a strictly larger ROA for OCT compared to TMPC as shown in Figure \ref{Fig:all_together}.

Figure \ref{Fig:all_together} also shows the computation times required to solve the online optimization problem for each controller, averaged over all feasible initializations used in the ROA analysis. It can be seen that the proposed method is an order of magnitude faster than the FPD approach, and requires a similar run time to TMPC. This is because the number of optimization variables and constraints for OCT and TMPC are the same, whereas FPD has a quadratic growth in the size of the optimization problem with increase in the prediction horizon. Thus, the proposed method is able to strictly improve on the ROA of the TMPC approach at no increase in computational cost.

\resp{
The closed-loop performances of TMPC and OCT methods are compared in Figure \ref{Fig:perf}. The closed-loop costs achieved by TMPC and OCT are averaged over 50 realizations of $w_k$ generated randomly such that Assumption \ref{As:disturbance_assumption} is satisfied. It can be seen that the costs are within 0.5\% of each other, and that OCT results in a larger feasible region as indicated by the red dots. 
}

\section{Conclusions and Outlook}\label{conclusions}
A novel algorithm is presented for the offline optimization of constraint tightening to be used in robust MPC controllers. The constraint tightenings are formulated as convex functions of a disturbance affine feedback gain. A convex program is solved to minimize the constraint tightening in order to increase the region of attraction. The proposed method guarantees recursive feasibility and input-to-state stability, and numerical examples demonstrate the computational efficiency and improved region of attraction compared to existing methods from the literature. One promising direction to improve the proposed approach is to design multiple feedback gains which can be fused online.

\bibliography{CDC22_aff}

\begin{thebibliography}{10}
\providecommand{\url}[1]{#1}
\csname url@rmstyle\endcsname
\providecommand{\newblock}{\relax}
\providecommand{\bibinfo}[2]{#2}
\providecommand\BIBentrySTDinterwordspacing{\spaceskip=0pt\relax}
\providecommand\BIBentryALTinterwordstretchfactor{4}
\providecommand\BIBentryALTinterwordspacing{\spaceskip=\fontdimen2\font plus
\BIBentryALTinterwordstretchfactor\fontdimen3\font minus
  \fontdimen4\font\relax}
\providecommand\BIBforeignlanguage[2]{{%
\expandafter\ifx\csname l@#1\endcsname\relax
\typeout{** WARNING: IEEEtran.bst: No hyphenation pattern has been}%
\typeout{** loaded for the language `#1'. Using the pattern for}%
\typeout{** the default language instead.}%
\else
\language=\csname l@#1\endcsname
\fi
#2}}

\bibitem{morariBook}
F.~Borrelli, A.~Bemporad, and M.~Morari, \emph{Predictive Control for Linear
  and Hybrid Systems}.\hskip 1em plus 0.5em minus 0.4em\relax Cambridge
  University Press, 2017.

\bibitem{Kouvaritakis}
B.~Kouvaritakis and M.~Cannon, ``Model predictive control: Classical, robust
  and stochastic,'' \emph{Advanced Textbooks in Control and Signal Processing},
  01 2016.

\bibitem{Chisci}
L.~Chisci, J.~Rossiter, and G.~Zappa, ``Systems with persistent disturbances:
  predictive control with restricted constraints,'' \emph{Automatica}, vol.~37,
  no.~7, pp. 1019--1028, 2001.

\bibitem{Mayne2}
D.~Mayne, M.~Seron, and S.~Rakovi{\'c}, ``Robust model predictive control of
  constrained linear systems with bounded disturbances,'' \emph{Automatica},
  vol.~41, no.~2, pp. 219--224, 2005.

\bibitem{Zanon}
M.~Zanon and S.~Gros, ``On the similarity between two popular tube {MPC}
  formulations,'' in \emph{European Control Conf.}, 2021, pp. 651--656.

\bibitem{munoz2014striped}
D.~Munoz-Carpintero, B.~Kouvaritakis, and M.~Cannon, ``Striped parameterized
  tube model predictive control,'' \emph{IFAC Proceedings Volumes}, vol.~47,
  no.~3, pp. 11\,998--12\,003, 2014.

\bibitem{homothetic}
S.~V. Rakovi{\'c}, B.~Kouvaritakis, R.~Findeisen, and M.~Cannon, ``Homothetic
  tube model predictive control,'' \emph{Automatica}, vol.~48, 2012.

\bibitem{Goulart}
P.~J. Goulart, E.~C. Kerrigan, and J.~M. Maciejowski, ``Optimization over state
  feedback policies for robust control with constraints,'' \emph{Automatica},
  vol.~42, no.~4, pp. 523--533, 2006.

\bibitem{zanon2020safe}
M.~Zanon and S.~Gros, ``Safe reinforcement learning using robust {MPC},''
  \emph{IEEE Trans. on Automatic Control}, vol.~66, no.~8, pp. 3638--3652,
  2020.

\bibitem{spacecraft}
M.~Mammarella, E.~Capello, H.~Park, G.~Guglieri, and M.~Romano, ``Tube-based
  robust model predictive control for spacecraft proximity operations in the
  presence of persistent disturbance,'' \emph{Aerospace Science and
  Technology}, vol.~77, pp. 585--594, 2018.

\bibitem{sieber2022system}
J.~Sieber, A.~Zanelli, S.~Bennani, and M.~N. Zeilinger, ``System level
  disturbance reachable sets and their application to tube-based {MPC},''
  \emph{European Journal of Control}, 2022.

\bibitem{Kogel}
M.~K{\"o}gel and R.~Findeisen, ``Fusing multiple time varying tubes for robust
  {MPC},'' \emph{IFAC-PapersOnLine}, vol.~53, no.~2, pp. 7055--7062, 2020, 21st
  IFAC World Congress.

\bibitem{boyd_vandenberghe_2004}
S.~Boyd and L.~Vandenberghe, \emph{Convex optimization}.\hskip 1em plus 0.5em
  minus 0.4em\relax Cambridge University Press, 2004.

\bibitem{iss}
Z.-P. Jiang and Y.~Wang, ``Input-to-state stability for discrete-time nonlinear
  systems,'' \emph{Automatica}, vol.~37, no.~6, pp. 857--869, 2001.

\bibitem{Khalil}
H.~K. Khalil, \emph{Nonlinear Systems, Third Edition}.\hskip 1em plus 0.5em
  minus 0.4em\relax Prentice Hall, Upper Saddle River, New Jersey 07458, 2002.

\bibitem{mosek}
{{MOSEK ApS}}, \emph{The MOSEK optimization toolbox for MATLAB manual. Version
  9.0.}, 2019.

\bibitem{yalmip}
J.~L{\"{o}}fberg, ``{YALMIP} : A toolbox for modeling and optimization in
  {MATLAB},'' in \emph{In Proceedings of the CACSD Conference}, Taipei, 2004.

\bibitem{MPT3}
M.~Herceg, M.~Kvasnica, C.~Jones, and M.~Morari, ``{Multi-Parametric Toolbox
  3.0},'' in \emph{European Control Conf.}, Z\"urich, 2013, pp. 502--510.

\bibitem{repo_ETH}
``{Supplemental material to the paper "Computationally efficient robust {MPC}
  using optimized constraint tightening"},'' DOI:
  \url{10.3929/ethz-b-000537515}, 2022, {ETH Research Collection}.

\end{thebibliography}

\end{document}